\documentclass[a4paper,12pt]{article}

\usepackage[pdftex]{graphicx}
\usepackage{amsthm}
\usepackage{amsmath}
\usepackage{epsfig}
\usepackage{amssymb}
\usepackage{latexsym}
\usepackage[T1]{fontenc}
\usepackage[utf8]{inputenc}
\usepackage{color}
\usepackage{url}

\usepackage{multirow}
\usepackage{dcolumn}
\newcolumntype{d}[1]{D{.}{.}{#1}}

\newtheorem{definition}{Definition}
\newtheorem{theorem}[definition]{Theorem}
\newtheorem{proposition}[definition]{Proposition}
\newtheorem{corollary}[definition]{Corollary}
\newtheorem{lemma}[definition]{Lemma}
\newtheorem{conjecture}{Conjecture}

\newtheorem*{mainlemma}{Parity Lemma}
\newcommand{\ka}{\color[rgb]{0,0,0}\ }

\let\epsilon=\varepsilon

\begin{document}

\title{\textbf{Small snarks with large oddness}}

\author{
Robert Lukoťka${}^1$,
Edita Máčajová${}^2$,
Ján Mazák${}^1$,
Martin Škoviera${}^2$
\\[3mm]
\\{\tt \{robert.lukotka, jan.mazak\}@truni.sk}
\\{\tt \{macajova, skoviera\}@dcs.fmph.uniba.sk}
\\[5mm]
${}^1$ Trnavská univerzita. Priemyselná 4, 918 43 Trnava\\
${}^2$ Univerzita Komenského, Mlynská dolina, 842 48 Bratislava
}

\maketitle

\begin{abstract}
We estimate the minimum number of vertices of a cubic graph
with given oddness and cyclic connectivity. We prove that a
bridgeless cubic graph $G$ with oddness $\omega(G)$ other than
the Petersen graph has at least $5.41\cdot\omega(G)$ vertices,
and for each integer $k$ with $2\le k\le 6$  we construct an
infinite family of cubic graphs with cyclic connectivity $k$
and small oddness ratio $|V(G)|/\omega(G)$. In particular, for
cyclic connectivity $2$, $4$, $5$, and $6$ we improve the upper
bounds on the oddness ratio of snarks to $7.5$, $13$, $25$, and
$99$ from the known values $9$, $15$, $76$, and $118$,
respectively. In addition, we construct a cyclically
$4$-connected snark of girth $5$ with oddness $4$ on $44$
vertices, improving the best previous value of $46$.
\end{abstract}

\section{Introduction}

Cubic graphs -- and more generally all graphs with maximum
degree $3$ -- naturally fall into two classes depending on
whether they do or do not admit a $3$-edge-colouring.
Accordingly, graphs from the former class are often called
\emph{colourable} while those in the latter class are called
\emph{uncolourable}. Although most cubic graphs are known to be
colourable \cite{robinson}, the problem of determining whether
a cubic graph has a $3$-edge-colouring is NP-complete
\cite{holyer}.

In the study of uncolourable graphs the following lemma plays a
fundamental role.
\begin{mainlemma}\label{paritylema}
Let $G$ be a cubic graph endowed with a proper
$3$-edge-colouring with colours $1$, $2$, and $3$. If a cutset
consisting of $m$ edges contains $m_i$ edges of colour $i$ for
$i\in\{1,2,3\}$, then
$$
m_1\equiv m_2\equiv m_3\equiv m\pmod 2.
$$
\end{mainlemma}
An immediate consequence of the Parity Lemma is that every
cubic graph with a bridge is uncolourable. Besides this trivial
family of uncolourable cubic graphs there exist more
interesting examples, called \emph{snarks}. These are
$2$-connected cubic uncolourable graphs, sometimes required to
satisfy additional conditions, such as cyclic
$4$-edge-connectivity and girth at least five, to avoid
triviality. Snarks are quintessential to many important
problems and conjectures in graph theory including the
$4$-colour theorem, Tutte's $5$-flow conjecture, the cycle
double cover conjecture, and many others. While most of these
problems are trivial for $3$-edge-colourable graphs, they are
exceedingly difficult for snarks in general. On the other hand,
for those which are close to being colourable they are usually
tractable. For example, the $5$-flow conjecture has been
verified for snarks with oddness at most $2$ (Jaeger
\cite{jaeger2}), and the cycle double cover conjecture has been
verified for snarks with oddness at most $4$ (Huck and Kochol
\cite{hk}, H\"aggkvist and McGuinness \cite{hg}). Snarks with
large oddness thus remain potential counterexamples to these
conjectures and therefore deserve further study.

Oddness is a natural measure of uncolourability of a cubic
graph based on the fact that every bridgeless cubic graph has a
$1$-factor \cite{petersen} and consequently also a $2$-factor.
It is easy to see that a cubic graph is $3$-edge-colourable if
and only if it has a $2$-factor that only consists of even
circuits. In other words, snarks are those cubic graphs which
have an odd circuit in every $2$-factor. The minimum number of
odd circuits in a $2$-factor of a bridgeless cubic graph $G$ is
its \emph{oddness}, and is denoted by $\omega(G)$. Since every
cubic graph has even number of vertices, its oddness must also
be even.

Another natural measure of uncolourability of a cubic graph is
based on minimising the use of the fourth colour in a
$4$-edge-colouring of a cubic graph. Alternatively, one can ask
how many edges have to be deleted in order to get a
$3$-edge-colourable graph. Somewhat surprisingly, the required
number of edges to be deleted is the same as the number of
vertices that have to be deleted in order to get a
$3$-edge-colourable graph (see \cite[Theorem~2.7]{steffen1}).
This quantity is called the \emph{resistance} of $G$, and will
be denoted by~$\rho(G)$. Observe that $\rho(G) \le \omega(G)$
for every bridgeless cubic graph $G$ since deleting one edge
from each odd circuit in a $2$-factor leaves a colourable
graph. On the other hand, the Parity Lemma implies that
$\rho(G)$ never equals $1$, which in turn yields that
$\rho(G)=2$ if and only if $\omega(G)=2$. The difference
between $\omega(G)$ and $\rho(G)$ can be arbitrarily large in
general \cite{steffen2}, nevertheless, resistance can serve as
a convenient lower bound for oddness, because it is somewhat
easier to handle.

The purpose of this article is to study how large the oddness
of a cubic graph can be compared to its order. The results of
this comparison strongly depend on the cyclic connectivity of
the graphs in question. We are therefore interested in the
smallest possible value of the ratio $|V(G)|/\omega(G)$ for a
snark $G$ within the class of cyclically $k$-connected snarks.
Recall that a cubic graph $G$ is \emph{cyclically
$k$-connected} if no set of fewer that $k$ edges can separate
two circuits of $G$ into distinct components; the \emph{cyclic
connectivity} $\zeta(G)$ of $G$ is the largest $k$ such that
$G$ is cyclically $k$-connected. It is useful to realise that
the values of vertex-connectivity, edge-connectivity, and
cyclic connectivity of a cubic graph coincide for $k\le 3$ but
cyclic connectivity may be arbitrarily large (see \cite{NScc}).

So far, only trivial lower bounds for the oddness ratio
$|V(G)|/\omega(G)$ have been known; as regards the upper
bounds, there are various constructions of snarks which are
probably not optimal. Since the oddness ratio of the Petersen
graph equals $5$, it is meaningless to attempt improving this
absolute lower bound. In Section~\ref{sec:asympt} we therefore
adopt an asymptotical approach similar to that taken by Steffen
\cite{steffen2} and by H\"agglung \cite{hagglund}. We summarise
the previously known results as well as our improvements in
Table~\ref{table}. We only consider cyclic connectivity $k
\in\{2,3,4,5,6\}$ because no cyclically $7$-connected snarks
are known. In fact, Jaeger \cite{jaeger7cc} conjectured that
such snarks do not exist. We prove, in particular, that the
oddness ratio of every snark is bounded above by $7.5$. We
conjecture that this general bound is best possible as we
believe that all snarks with oddness $\omega$ have at least
$7.5 \omega -5$ vertices (see Section~\ref{sec:2} for details).

\begin{table}[h]
\label{table} \caption{Upper and lower bounds on oddnes ratio
$|V|/\omega$.}
\begin{center}
\begin{tabular}{|c||d{1.2}|d{2.1}|rl|}
\hline
connectivity $k$
& \multicolumn{1}{c}{lower bound}
&\multicolumn{1}{|c}{current upper bound}
& \multicolumn{2}{|c|}{previous upper bound}\\
\hline\hline
2 & 5.41 &  7.5& 9 & (Steffen \cite{steffen2})     \\\hline
3 & 5.52 &    9& 9 & (Steffen \cite{steffen2})     \\\hline
4 & 5.52 & 13&  15 & (H\"agglund \cite{hagglund})  \\\hline
5 & 5.83 &   25&76 & (Steffen \cite{steffen2})      \\\hline
6 & 7    &  99&118 & (Kochol \cite{kocholOddness}) \\\hline
\end{tabular}
\end{center}
\end{table}

Besides general bounds, we are also interested in identifying
the smallest snark with oddness $4$, addressing a long-standing
open problem restated as Problem 4 in \cite{brinkmann}. We show
that the smallest order of a snark with oddness $4$ is $28$ and
present one with cyclic connectivity $3$ (see
Figure~\ref{fig:snark28} right). We further construct a
cyclically $4$-connected snark of girth $5$ with oddness $4$
and resistance $3$ on $44$ vertices (see Figure~\ref{snark44}),
improving by $2$ the value established in \cite{hagglund}. We
believe that $44$ is the smallest possible order of a
non-trivial snark of oddness~$4$.

\section{Oddness and resistance ratios}\label{sec:asympt}

The \emph{oddness ratio} of a snark $G$ is the quantity
$|V(G)|/\omega(G)$, and its \emph{resistance ratio} is the
quantity $|V(G)|/\rho(G)$. In order to derive some relevant
information about these parameters we also examine their
asymptotic behaviour. To this end, we define
$$
A_\omega = \liminf_{|V(G)|\to\infty} {|V(G)|\over \omega(G)}
$$
and
$$
A_\rho = \liminf_{|V(G)|\to\infty} {|V(G)|\over \rho(G)}.
$$
Since the oddness ratio is at least as large as its resistance
ratio, we have $A_\omega\le A_\rho$.

The oddness and resistance ratios heavily depend on the cyclic
connectivity of a graph in question. This suggests to study
analogous values $A_\omega^k$ and~$A_\rho^k$ obtained under the
assumption that the class of snarks is restricted to those with
cyclic connectivity at least~$k$.  Note that $A_\omega^k\le
A_\rho^k$ for every $k\ge 2$ and that $A_\omega^2=A_\omega$ and
$A_\rho^2=A_\rho$.

Similar ideas were pursued by Steffen \cite{steffen2} who asked
the following question: For which values $a_\rho$ the
inequality $n\ge a_\rho\rho(G)$ has only finitely many
counterexamples? He proved (Theorem 2.4 and Lemma 3.4 in
\cite{steffen2}) that the resistance ratio of every snark of
order $n>14$ is at least $8$ and constructed a family of snarks
for which $\rho(G)\ge n/9$. It follows that $8\le A_\rho\le 9$
and therefore $A_\omega\le A_\rho\le 9$. Since the snarks that
he constructed are cyclically $3$-connected, we also have
$A_\omega=A_\omega^2\le A_\omega^3\le A_\rho^3\le 9$. An
earlier construction of Rosenfeld \cite{rosenfeld} provided the
bounds $A_\omega\le A_\rho\le 10$.

In this paper we concentrate on upper and lower bounds for
$A_\omega$ and $A_\omega^k$, while their resistance
counterparts $A_\rho$ and $A_\rho^k$ will only play an
auxiliary role. To start with, let us discuss the value
$A_\omega^6$. Since every odd circuit in a cyclically
$6$-connected snark has length at least $7$, every such snark
of oddness $\omega$ has at least $7\omega$ vertices. Thus
$A_\omega^6\ge 7$, which fills in the entry in the second
column of the last line of Table~\ref{table}.

\section{Reduction lemmas}\label{sec:reduction}

In the study of various properties of snarks it is often
convenient to avoid snarks having short circuits or small
edge-cuts, because such snarks can either be considered trivial
or lack the desired properties. There are well-known reductions
that remove these structures from snarks such as contraction of
a 3-circuit or a 2-circuit and suppression of a $2$-valent
vertex  whenever it arises \cite{watkins2}. Similar reductions
will also be needed in our further investigation.

We begin with the observation that in the search for small
snarks with large oddness we can ignore graphs with parallel
edges or triangles.

\begin{lemma}\label{lemma:girth4}
For every snark $G$ there exists a snark $G'$ of order not
exceeding that of $G$ such that $\omega(G')=\omega(G)$,
$\zeta(G')\ge\zeta(G)$, and the girth of $G'$ is at least $4$.
\end{lemma}

\begin{proof}
Let $G_0$ be the snark obtained from $G$ by the standard
reduction of a single circuit $C$ of length $2$ or $3$: the
graph $G_0$ arises from $G$ simply by the contraction of $C$
into a single vertex and by suppressing a vertex of degree two,
if it arises. It is immediate that $\zeta(G_0)\ge \zeta(G)$.

For each $2$-factor $F_0$ of $G_0$ there exists a corresponding
$2$-factor $F$ in $G$ such that every circuit of $F_0$ extends
to a unique circuit of $F$ with the same parity of the number
of vertices and at least the same length. Since $F$ contains no
circuits other than those corresponding to the circuits of
$F_0$, we have $\omega(G_0)\ge \omega(G)$.  On the other hand,
$G$ always contains a $2$-factor $F$ that contains no triangles
\cite{malerezy}. This $2$-factor has a corresponding $2$-factor
$F_0$ in $G_0$ with the same number of odd circuits, and thus
$\omega(G_0)\le \omega(G)$.

The required graph $G'$ is now obtained by repeating a similar
procedure with any circuit of length $2$ or $3$ in $G_0$.
\end{proof}

The standard reduction of a $4$-circuit, which consists in
removing a pair of opposite edges and suppressing the resulting
vertices of degree $2$, does not work here because it may
decrease cyclic connectivity, produce a bridge, or even create
a disconnected graph. Nevertheless, for our purpose the
following weaker result is sufficient.

\begin{lemma}\label{lemma:girth5}
For every snark $G$ there exists a snark $G'$ of order not
exceeding that of $G$ such that $\omega(G')=\omega(G)$ and the
girth of $G'$ is at least $5$.
\end{lemma}

\begin{proof}
For every $4$-circuit in $G$ there are two possibilities for
applying the standard reduction which correspond to two pairs
of opposite edges in a quadrilateral. It is not difficult to
see that one of the choices always creates a $2$-connected
graph \cite{watkins2}. Furthermore, any $2$-factor of the
reduced graph can be easily extended to a $2$-factor of the
original graph without changing the parity of the lengths of
the circuits and by adding at most one $4$-circuit.
Consequently, after reducing a $4$-cycle in a snark we obtain a
smaller snark with the same oddness. We repeatedly apply
reductions of circuits of length $4$ and, if necessary, we also
reduce circuits of length $2$ and $3$ in the manner described
in the previous proof.  Eventually we obtain a snark $G'$ which
has the same oddness as $G$, girth at least $5$, and is
$2$-connected.
\end{proof}

The remaining two lemmas deal with small edge-cuts in snarks.

\begin{lemma}\label{lemma:2-edge-cut-uncolourable}
For every snark $G$ there exists a snark $G'$ of order not
exceeding that of $G$ such that $\omega(G')=\omega(G)$ and
every $2$-edge-cut in $G'$ separates two uncolourable subgraphs
of $G'$.
\end{lemma}

\begin{proof}
Let $G$ be a snark with a $2$-edge-cut $S$ that separates $G$
into two components one of which is colourable. If both
components were colourable, then so would be $G$, by the Parity
Lemma. It follows that one of the components  is colourable and
the other is not. Let $G_1$ be the uncolourable component, and
let $G_2$ be the colourable one. For $i\in \{ 1, 2 \}$, let
$G_i'$ arise from $G_i$ by joining its vertices of degree two
by an edge $e_i$. By the Parity Lemma, $G_2'$ is colourable, so
there is an even $2$-factor in $G_2'$ passing through~$e_2$.
Consequently, every $2$-factor $F_1$ of $G_1'$ can be extended
to a $2$-factor $F$ of $G$ such that no odd circuit of $F$ is
contained $G_2$. Since $G_2$ has even number of vertices, the
oddness of $G$ does not exceed the that of $G_1'$. By repeating
the procedure with $G_1'$ we eventually obtain the desired
graph $G'$.
\end{proof}

Similar arguments can be used to prove an analogous result for
$3$-edge-cuts.

\begin{lemma}\label{lemma:3-edge-cut-uncolourable}
For every snark $G$ there exists a snark $G'$ of order not
exceeding that of $G$ such that $\omega(G')=\omega(G)$ and
every $3$-edge-cut in $G'$ separates two uncolourable subgraphs
of $G'$.
\end{lemma}

\section{Lower bounds on oddness ratio}\label{sec:lowerbd}

Lemma~\ref{lemma:girth4} implies that the oddness ratio of
every snark is at least $5$. This bound is best possible,
because the Petersen graph has ten vertices and oddness~$2$.
However, the Petersen graph is the only snark for which
equality holds, as it is the only $2$-edge-connected cubic
graph having only $5$-circuits in each $2$-factor~\cite{devos}.
The purpose of this section is to improve this bound for snarks
different from the Petersen graph.

We begin with the key observation that in a cubic graph one can
always find a $2$-factor that avoids most of the chosen
$5$-circuits.

\begin{proposition}\label{prop:basic}
Let $\mathcal{C}$ be a set of $5$-circuits of a bridgeless
cubic graph $G$. Then $G$ has a $2$-factor that contains at
most $1/6$ of $5$-circuits from $\mathcal{C}$.
\end{proposition}

In the proof of this lemma we employ the concept of a perfect
matching polytope introduced in \cite{edmonds}. Let $G$ be a
graph with $E(G)=\{e_1,e_2,\dots,e_m\}$. Each perfect matching
$M$ of $G$ can be represented by its characteristic vector
$\mathbf{x}\in\mathbb{R}^m$ in which the $i$-th entry is $1$ if
$e_i$ belongs to $M$ and is $0$ otherwise. The \emph{perfect
matching polytope} $\mathcal{P}(G)$ of a graph $G$ is the
convex hull in $\mathbb{R}^m$ of the set of characteristic
vectors of all perfect matchings of $G$.

Let us denote the entry of a vector $\mathbf{x}\in\mathbb{R}^m$
corresponding to an edge $e\in E(G)$ by $\mathbf{x}(e)$, and
let
$$\mathbf{x}(S)=\sum_{e\in S}\mathbf{x}(e)$$
whenever $S$ is a subset of $E(G)$. For a set of vertices or a
subgraph $U$ of $G$ let $\delta(U)$ denote the set of all edges
with precisely one end in $U$. With this notation we can
equivalently describe $\mathcal{P}(G)$ as the set of all
vectors from $\mathbb{R}^m$ that satisfy the following
inequalities:
\begin{eqnarray*}
\mathbf{x}(e) &\ge& 0 \qquad \text{for each } e\in E(G),\\
\mathbf{x}(\delta(v)) &=& 1 \qquad \text{for each } v\in V(G),\\
\mathbf{x}(\delta(U)) &\ge& 1 \qquad \text{for each } U
\subseteq V(G)\ \hbox{with}\ |U|\ \hbox{odd}.\\
\end{eqnarray*}
Note that if $G$ is cubic and bridgeless, the vector $(1/3,
1/3, \dots, 1/3)$ always belongs to $\mathcal{P}(G)$.

\begin{proof}[Proof of Proposition~\ref{prop:basic}]
Let $\mathcal{P}$ be the perfect matching polytope of $G$;
since $G$ is cubic and bridgeless, $\mathcal{P}$ is nonempty.
Consider the function
$$
f(\mathbf{x})=\sum_{C\in\mathcal{C}}\, \sum_{e\in\delta(C)}\mathbf{x}(e)
$$
defined for each $\mathbf{x}\in \mathcal{P}$. The function $f$
is linear, hence there is a vector $\mathbf{x}_0$ such that
$f(\mathbf{x}_0)$ is minimal and $\mathbf{x}_0$ is a vertex of
$\mathcal{P}$. Since $(1/3,1/3,\dots,1/3)\in \mathcal{P}$, we
have $f(\mathbf{x}_0)\le 5/3\cdot |\mathcal{C}|$.

Let $M$ be the perfect matching corresponding to $\mathbf{x}_0$
and let $F$ be the $2$-factor complementary to $M$. Assume that
$F$ contains $k$ circuits from $\mathcal{C}$. If a $5$-circuit
$C\in \mathcal{C}$ belongs to~$F$, it adds $5$ to the sum in
$f(\mathbf{x}_0)$. If $C$ does not belong to~$F$, it adds at
least~$1$. Altogether $f(\mathbf{x}_0)\ge 5k+(|\mathcal{C}|-k)
= |\mathcal{C}|+4k$. Since $f(\mathbf{x}_0)\le 5/3\cdot
|\mathcal{C}|$, we obtain $k\le |\mathcal{C}|/6$.
\end{proof}

By inspecting all the vertices at distance at most $2$ from a
given vertex one can see that every vertex of a cubic graph is
contained in at most six $5$-circuits. In fact, the Petersen
graph is the only cubic graph that has a vertex contained in
precisely six $5$-circuits. As the Petersen graph is
vertex-transitive, it follows that it contains precisely twelve
$5$-circuits.

Consider the set $\mathcal{C}$ of all $5$-circuits of the
Petersen graph. Since every $2$-factor of the Petersen graph
consists of two $5$-circuits \cite{devos}, the constant $1/6$
in Proposition~\ref{prop:basic} is best possible.

\begin{corollary}\label{cor:avoid}
Let $G$ be a snark different from the Petersen graph. For every
vertex $v$ of~$G$ there exists a $2$-factor $F$ of $G$ such
that every $5$-circuit of $F$ misses $v$.
\end{corollary}

\begin{proof}
Let $\mathcal{C}$ be the set of all $5$-circuits passing
through a given vertex $v$ of $G$. Since $G$ is not the
Petersen graph, $|\mathcal{C}|\le 5$. By
Proposition~\ref{prop:basic}, there is a $2$-factor $F$ that
contains at most $5/6<1$ circuits from $\mathcal{C}$. Thus $F$
contains no $5$-circuit passing through $v$.
\end{proof}

To prove the main result of this section we need several
lemmas.

\begin{lemma}\label{lemma:5circuits}
Let $G$ be a snark of order $n$ different from the Petersen
graph. Assume that $G$ has $n_i$ vertices contained in exactly
$i$ $5$-circuits. Then the following are true:
\begin{itemize}
\item[{\rm (i)}] $n_6 = 0$
\item[{\rm (ii)}] $n_5 \le 2n/5$
\item[{\rm (iii)}] If $G$ is cyclically $3$-connected, then
    $n_5=0$.
\item[{\rm (iv)}] If $G$ is cyclically $5$-connected, then
    $n_4=0$.
\end{itemize}
\end{lemma}

\begin{proof}
We have already discussed the statement (i) before
Corollary~\ref{cor:avoid}. We give a detailed proof only for
(ii); the remaining statements follow by similar
considerations.

Let $v$ be a vertex of $G$ contained in five $5$-circuits. It
follows that one of the edges incident with $v$ belongs to four
of the $5$-circuits and the other two edges belong to three of
them. In particular, $v$ cannot be contained in a triangle.

Consider the subgraph $H$ of $G$ induced by the set of all
vertices at distance $2$ from $v$. Each $5$-circuit passing
through $v$ contributes one edge to $H$, so $H$ has five edges
and at most six vertices. If $v$ belonged to a $4$-circuit,
there would be at most $5$ vertices in $H$ and at least one of
them would be of degree at most $1$ in $H$. But then $H$ could
not have five edges. So $v$ belongs to no $4$-circuit and $H$
has six vertices.

Let $K$ be the subgraph of $G$ induced by all vertices at
distance at most $2$ from $v$. Since $H$ has five edges, there
are two edges separating $K$ from the rest of $G$. The subgraph
$K$ has ten vertices, and a short case analysis shows that $K$
contains at most three other vertices contained in five
$5$-circuits in~$G$. Hence, at most four of the ten vertices of
$K$ are contained in five $5$-circuits. By repeating this
argument for any vertex $v$ contained in five $5$-circuits
which we have not counted yet we arrive at the conclusion that
$n_5\le 4n/10=2n/5$. Note that no edge of $K$ belongs to a
$2$-edge-cut, therefore all the subgraphs $K$ arising in the
described way are pairwise disjoint.
\end{proof}

\begin{lemma}\label{lemma:doublecounting}
Let $G$ be a bridgeless cubic graph of order $n$ that has $n_i$
vertices contained in exactly $i$ $5$-circuits. Let
$\mathcal{C}$ be the set of all $5$-circuits of $G$. Then
$$
|\mathcal{C}|={1\over 5}\sum_{i=0}^6 i\cdot n_i.
$$
\end{lemma}

\begin{proof}
The desired equality immediately follows from counting the
number of pairs $(v, C)$, where $C\in \mathcal{C}$ and $v$ is a
vertex contained $C$, in two ways.
\end{proof}

\begin{lemma}\label{lemma:bound}
Let $G$ be a snark of order $n$ with girth at least $4$. If $G$
has $q$ circuits of length~$5$, then
$$
\omega(G)\le \frac{3n+q}{21}.
$$
\end{lemma}

\begin{proof}
Let $\omega$ be the oddness of $G$. By
Proposition~\ref{prop:basic}, there is a $2$-factor $F$ of $G$
which contains at most $q/6$ circuits of length $5$. Let $p_i$
denote the number of $i$-circuits in $F$ and let $s =
p_7+p_9+\dots$. Note that $p_5\le q/6$. We are minimising the
value of $n/\omega$ subject to the constraints
\begin{eqnarray*}
n &=& 4p_4+5p_5+6p_6+7p_7+\dots,\\
\omega &=& p_5+p_7+p_9+\dots\ .\\
\end{eqnarray*}
Roughly speaking, the ``worst case'' occurs when $F$ contains
no circuits of even length, there is a maximal possible number
of $5$-circuits, and all the longer odd circuits have length
$7$. In other words,
\begin{eqnarray*}
{n\over \omega}
& = & {4p_4+5p_5+6p_6+7p_7+\dots\over p_5+p_7+p_9+\dots}
      \ge {5p_5+7s\over p_5+s}=5+{2s\over p_5+s}\ge 5+{2s\over{1\over 6}q+s}\\
& = & 5 +{12\over {q\over s}+6}=5+{12\over {q\over\omega-p_5}+6}
      \ge 5+{12\over {6q\over 6\omega-q}+6}= 5+ {{6\omega-q}\over{3\omega}}.
\end{eqnarray*}
The resulting inequality is equivalent to the desired one.
\end{proof}

\begin{theorem}\label{thm:lowerbd}
Let $G$ be a snark of order $n$ different from the Petersen
graph. Then:
\begin{itemize}
\item[{\rm (i)}]  $n/\omega(G)\ge 525/97 > 5.41$;

\item[{\rm (ii)}] $n/\omega(G)\ge 105/19 > 5.52$, if $G$ is
    cyclically $3$-connected;

\item[{\rm (iii)}] $n/\omega(G)\ge 35/6 > 5.83$, if $G$ is
    cyclically $5$-connected.
\end{itemize}
\end{theorem}

\begin{proof}
According to Lemma \ref{lemma:girth4}, it is enough to prove
the theorem provided that $G$ has girth at least $4$.
Lemma~\ref{lemma:5circuits} asserts that each vertex of $G$ is
contained in at most five $5$-circuits and that at most $2n/5$
vertices are contained in five $5$-circuits. Let $\mathcal{C}$
be the set of all $5$-circuits of $G$. By using
Lemma~\ref{lemma:doublecounting} we have
$$
|\mathcal{C}|\le {1\over 5}(5n_5+4(n-n_5))=
{4\over 5}n+{1\over 5}n_5\le {22\over 25}n,
$$
and Lemma~\ref{lemma:bound} in turn yields that $n/\omega(G)\ge
525/97>5.41$. This proves (i).

To prove (ii), let $G$ be cyclically $3$-connected. Again, by
Lemma~\ref{lemma:girth4} we may assume $G$ to have girth at
least $4$. Let $\mathcal{C}$ be the set of all $5$-circuits of
$G$. From Lemma~\ref{lemma:5circuits} we infer that each vertex
of $G$ is contained in at most four $5$-circuits.  According to
Lemma~\ref{lemma:doublecounting} we  have $|\mathcal{C}| \le
{4\over 5}n$ and from Lemma~\ref{lemma:bound} we get
$n/\omega(G)\ge 105/19$, as claimed.

Finally, if $G$ is cyclically $5$-connected, each vertex of $G$
is contained in at most three $5$-circuits, by Lemma
\ref{lemma:5circuits}. According to
Lemma~\ref{lemma:doublecounting}, $|\mathcal{C}| \le {3\over
5}n$, and from Lemma~\ref{lemma:bound} we get $n/\omega(G)\ge
105/18 = 35/6$, which proves (iii).
\end{proof}

\section{Construction methods}\label{sec:blocks}

In the next few sections we construct snarks with small size
when compared to their oddness. Most of our constructions
produce larger graphs by putting smaller parts together. For
this purpose the following terminology will be useful: A
\emph{network} is a pair $(G,T)$ consisting of a graph $G$ and
a distinguished set $T$ of vertices of degree $1$  called
\emph{terminals}. An edge incident with a terminal is called a
\emph{terminal edge}. A network with $k$ terminals will be
called a $k$-\emph{pole}. In all networks considered below
nonterminal vertices will always have degree $3$.

Each terminal of a network serves as a place of connection with
another terminal. Two terminal edges of either the same network
or of two disjoint networks can be naturally joined to form a
new nonterminal edge by identifying the corresponding terminal
vertices and suppressing the resulting $2$-valent vertex. This
operation is called the \emph{junction} of two terminals.

A standard way to create terminal vertices in a graph or
network is by \emph{splitting off} a vertex $v$ from a
graph~$G$; by this we mean the removal of $v$ from $G$ and
attaching a terminal vertex to each dangling edge originally
incident with $v$. The terminals resulting from splitting off
the vertex $v$ from $G$ will be said to be \emph{corresponding
to} $v$.

The basic building blocks for our constructions are the
following six networks obtained from the Petersen graph. All of
them are displayed in Figure~\ref{fig:pet} together with their
simplified diagrams used in the rest of the paper.

\begin{figure}
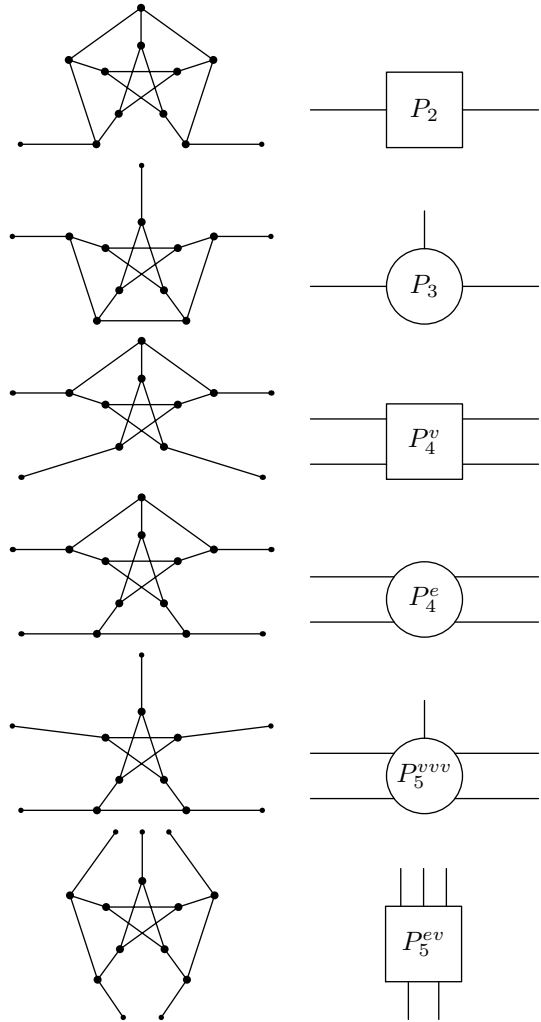

\center
\begin{tabular}{ccc}
\includegraphics{multipoles-1} & \includegraphics{schemmultip-1}\\
\includegraphics{multipoles-2} & \includegraphics{schemmultip-2}\\
\includegraphics{multipoles-3} & \includegraphics{schemmultip-3}\\
\includegraphics{multipoles-4} & \includegraphics{schemmultip-4}\\
\includegraphics{multipoles-5} & \includegraphics{schemmultip-5}\\
\includegraphics{multipoles-6} & \includegraphics{schemmultip-6}\\
\end{tabular}
\caption{Networks $P_2$, $P_3$, $P_4^v$, $P_4^e$, $P_5^{vvv}$, $P_5^{ev}$
and their diagrams.}
\label{fig:pet}
\end{figure}

\begin{itemize}
\item The network  $P_2$ is formed from the Petersen graph
    by subdividing an edge and splitting off the new
    vertex. This network has ten nonterminal vertices. By
    the Parity Lemma, it is not 3-edge-colourable.
\item The network $P_3$ is formed from $P$ by splitting off
    a vertex. This network has nine nonterminal vertices
    and is not $3$-edge-colourable for similar reasons.
\item The network $P_4^v$ is formed from $P$ by deleting an
    edge and splitting off the two resulting vertices of
    degree $2$. The terminal edges constitute two pairs,
    each pair corresponding to one vertex of degree two.
    Every $3$-edge-colouring of $P_4^v$ must assign the
    same colour to both edges within the same pair,
    otherwise one could properly colour the edges of $P$
    with three colours. This network has eight nonterminal
    vertices.
\item The network $P_4^e$ is formed from $P$ by subdividing
    two edges at distance $1$ and splitting off the
    resulting vertices of degree $2$. The terminal edges
    again constitute two pairs, each pair corresponding to
    one vertex of degree two. Every $3$-edge-colouring of
    $P_4^e$ must assign different colours to the edges in
    the same pair of terminal edges, otherwise one could
    properly colour the edges of $P$ with three colours.
    This network has ten nonterminal vertices.
\item The network $P_5^{vvv}$ is formed from $P$ by
    deleting two adjacent edges and splitting off the two
    resulting vertices of degree $2$. The terminal edges
    constitute two pairs, each pair corresponding to one
    vertex of degree two, and one single terminal edge.
    Since there are five terminal edges, the Parity Lemma
    implies that in every $3$-edge-colouring of $P_5^{vvv}$
    one colour must be used exactly three times and the
    other two colours must be used once. As we cannot
    colour $P$, every $3$-edge-colouring of $P_5^{vvv}$
    must assign the same colours to the terminal edges of
    one of the two pairs. The other pair of terminal edges
    must have different colours assigned. This network has
    seven nonterminal vertices.
\item The network $P_5^{ev}$ is formed from $P$ by
    splitting of a vertex $v$ and subdividing an edge at
    distance two from $v$ and by subsequently splitting off
    the new vertex of degree $2$. This creates one pair and
    one triple of terminal edges, respectively. Since $P$
    is not $3$-edge colourable, the Parity Lemma implies
    that the terminal edges contained in the pair must have
    different colours. This network has nine nonterminal
    vertices.
\end{itemize}

Another construction method which we employ is superposition
\cite{superposition}. Given a cubic graph $G$, take a
collection $\{X_v;\, v\in V(G)\}$ of disjoint networks called
\emph{supervertices} and a collection $\{Y_e;\, e\in E(G)\}$ of
disjoint networks called \emph{superedges}. Each supervertex is
a network whose terminals are partitioned into three subsets
and each superedge is a network whose terminals are partitioned
into two subsets; the partition sets are called
\emph{connectors}. For each vertex $v$ of $G$, associate each
connector of the supervertex $X_v$ with the end of an edge of
$G$ incident with $v$ in such a way that no two connectors are
associated with the same end. For each edge $e$ of $G$,
associate each connector of the superedge $Y_e$ with an end of
the edge $e$ in such a way that the connectors corresponding to
an incidence between a vertex and an edge in $G$ have the same
size; again, the connectors of $Y_e$ are associated with
different ends of $e$. Now, substitute each vertex $v$ of $G$
with the supervertex $X_v$ and each edge $e$ of $G$ with the
superedge $Y_e$, and perform all the junctions between
supervertices and superedges that correspond to the incidences
between vertices and edges of $G$. Let $\tilde G$ be the
resulting graph. Note that the graph $\tilde G$ is cubic since
all the nonterminal vertices in supervertices and superedges
have degree $3$. We call $\tilde G$ a \textit{superposition} of
$G$.

There is a natural incidence-preserving surjective mapping
$p\colon \tilde G\to G$, called a \textit{projection}, such
that for each vertex $\tilde v$ of $\tilde G$ the image
$p(\tilde v)$ is a vertex of $G$ and for each edge $\tilde e$
of $\tilde G$ the image $p(\tilde e)$ is either an edge of $G$
or a vertex resulting from the contraction of~$\tilde e$. This
mapping takes every nonterminal vertex from a supervertex $X_v$
to $v$ itself, and every nonterminal vertex from a superedge
$X_e$ that replaces an edge $e=uv$ of $G$ to either $u$ or~$v$.
It should be remarked that $p$ is not uniquely determined, but
the difference between any two such mappings is insubstantial.

Note that in a superposition $\tilde G$ of $G$ a vertex $v$ of
$G$ may be substituted with a \textit{trivial supervertex}
$X_v$, one which consists of a single nonterminal vertex
$\tilde v$ incident with three pendant edges whose pendant
vertices are terminals. Similarly, an edge $e$ may be
substituted with a \textit{trivial superedge} $Y_e$ which
consists of a single edge whose endvertices are both terminals.
If all the substitutions are trivial, then $p\colon \tilde G\to
G$ is clearly an isomorphism.

A superedge $Y$ is \emph{proper} if for every
$3$-edge-colouring of $Y$ that uses nonzero elements of
$\mathbb{Z}_2 \times \mathbb{Z}_2$ as colours the sum of
colours on the terminal edges incident with the terminals in
either connector is nonzero. In particular, a trivial superedge
is proper. It follows from the Parity Lemma that the two sums
must be the same element of $\mathbb{Z}_2 \times \mathbb{Z}_2$.

There is a standard method of constructing proper superedges
from a snark. Take a snark $H$, and choose a vertex $v$ which
can be either a vertex of $H$ or a $2$-valent vertex arising
from the subdivision of an arbitrary edge of $H$. Let us create
a connector by splitting off $v$ from $H$ and repeat the
operation with another such vertex. Clearly, splitting off a
vertex of $H$ creates a connector of size $3$ whereas splitting
off a subdivision vertex creates a connector of size $2$. It is
an easy consequence of the Parity Lemma that any superedge
arising in the just described way is always proper.

A superposition $\tilde G$ of $G$ is \emph{proper} if $G$ is a
snark and  every superedge is proper. We claim that in this
case $\tilde G$ cannot be $3$-edge-colourable. Suppose it is,
and consider a $3$-edge-colouring of $\tilde G$ that uses
nonzero elements of $\mathbb{Z}_2 \times \mathbb{Z}_2$ as
colours. We show that the projection $p\colon \tilde G\to G$
induces a $3$-edge-colouring of $G$.  We can colour each edge
$e$ of $G$ with the element of $\mathbb{Z}_2 \times
\mathbb{Z}_2$ equal to the sum of colours occurring on the
terminal edges of either connector of $Y_e$. This colouring is
easily seen to be a proper $3$-edge-colouring of~$G$, which is
impossible since $G$ is a snark. It follows that $\tilde G$ is
a snark.

The following more general result is useful for constructing
graphs with large resistance or oddness.

\begin{proposition}\label{prop:superresist}
Let $\tilde G$ be a snark resulting from a proper superposition
of a snark $G$. Then $\rho(\tilde G)\ge \rho(G)$.
\end{proposition}

\begin{proof}
Let $W$ be a set of vertices of $\tilde G$ such that $\tilde
G-W$ is a $3$-edge-colourable graph and $|W|=\rho(\tilde G)$ .
Consider the image $p(W)$ under the projection $p\colon\tilde
G\to G$. The restriction of $p$ to $\tilde G-W$ maps $\tilde
G-W$ to $G-p(W)$ and induces a $3$-edge-colouring of $G-p(W)$
from any $3$-edge-colouring of $\tilde G-W$ in the manner
described above. Therefore
$$\rho(\tilde G)=|W|\ge |p(W)|\ge \rho(G),$$
as claimed.
\end{proof}

It is not known whether the statement about oddness analogous
to Proposition~\ref{prop:superresist} is true.

\section{Connectivity 2}\label{sec:2}

In this section we identify the smallest snark with oddness $4$
and construct a family of snarks with oddness $2q$ having less
than $15q$ vertices.

By virtue of Lemma~\ref{lemma:girth5}, we only need to consider
snarks of girth at least $5$. Since cyclically $4$-connected
snarks are catalogued up to order $36$ and have oddness two up
to this order \cite{brinkmann}, we may restrict to graphs that
either have a $2$-edge-cut or a $3$-edge-cut.

\begin{figure}
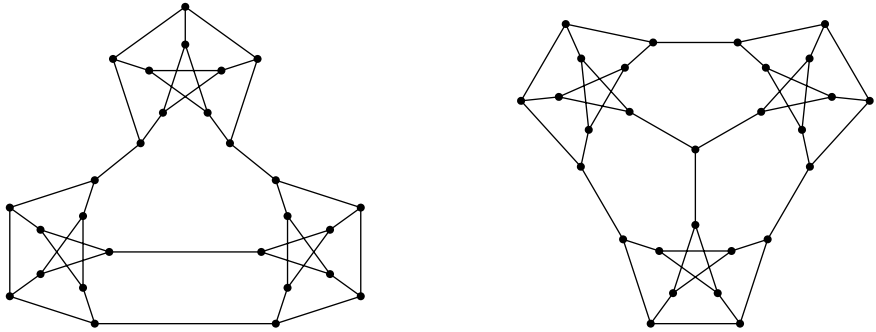

\center
\includegraphics{obrazok3-1}\hskip 2cm
\includegraphics{obrazok2-1}
\caption{The smallest snarks of oddness $4$; both have order $28$.}
\label{fig:snark28}
\end{figure}

\begin{theorem}\label{thm:3connected}
The smallest snark with oddness $4$ has $28$ vertices. There is
one such snark with cyclic connectivity $2$ and one with cyclic
connectivity $3$.
\end{theorem}

\begin{proof}
Figure~\ref{fig:snark28} shows two snarks of order 28, $H_1$
and $H_2$, with cyclic connectivity $2$ and $3$, respectively.
Each of them contains three disjoint copies of $P_3$, which is
uncolourable, therefore $\rho(H_i)\ge 3$, and hence
$\omega(H_i)\ge 4$. It is not difficult to show that in fact
$\omega(H_1)=\omega(H_2)=4$. In the rest of the proof we show
that there exists no snark with oddness $4$ on fewer than~$28$
vertices.

Let $G$ be a snark with oddness $4$ of minimum order, and
suppose that its order is at most~$26$. By
Lemmas~\ref{lemma:girth4} and~\ref{lemma:girth5} its girth is
at least $5$.\ka It is known \cite{brinkmann} that all
cyclically $4$-connected snarks with girth $5$ or more and
order not exceeding $36$ have oddness~$2$. Hence the cyclic
connectivity of $G$ equals either $2$ or $3$.

First, let us suppose that $G$ is cyclically $3$-connected.
Since $G$ is not cyclically $4$-connected, it contains a
cycle-separating $3$-edge-cut $S$. By
Lemma~\ref{lemma:3-edge-cut-uncolourable}, the edge-cut $S$
separates two uncolourable subgraphs $G_1$ and $G_2$. Let
$G_i'$ be the snark obtained from $G_i$ by joining a new vertex
$v_i$ to the three vertices of degree $2$ in $G_i'$. The snark
$G_i'$ is cyclically $3$-connected because $G$ was cyclically
$3$-connected. In other words, $G$ arises from two smaller
cyclically $3$-connected snarks $G_1'$ and $G_2'$ by the
reverse process, which splits off a vertex from both $G_1'$ and
$G_2'$, and joins the terminals from different snarks. With the
help of a computer we have checked that there is no snark $G$
with oddness $4$ on at most $26$ vertices that arises in this
way.

It follows that $G$ has a $2$-edge-cut. Let $S$ be a
$2$-edge-cut in $G$ such that one of the components separated
by $S$ is as small as possible. Again, we may assume that $G$
has girth at least $5$. Further, by
Lemma~\ref{lemma:2-edge-cut-uncolourable}, the cut $S$
separates two uncolourable subgraphs $G_1$ and~$G_2$, both with
even number of vertices. Both subgraphs must have at least ten
vertices, for otherwise we would obtain a snark of order at
most eight by connecting the two vertices of either $G_1$ or
$G_2$ with an edge, but there are no such snarks. Thus the
larger of the two subgraphs has at most $16$ vertices. Let
$G_i'$ be the snark obtained from $G_i$ by joining the vertices
of degree $2$ by an edge $e_i$. We may assume that $10 \le
|V(G_1')| \le |V(G_2')|\le 16$, hence $G_1'$ has $10$ or $12$
vertices.

If both $G_1'$ and $G_2'$ have a $2$-factor such that the added
edge $e_i$ belongs to an odd circuit of this $2$-factor, then
there $G$ has a $2$-factor $F$ containing an even circuit $C$
that passes through $S$. Let $F_i$ be the $2$-factor of $G_i'$
obtained by adding $e_i$ to $F \cap G_i$. The circuit $C$ has
length at least $6$ because $G$ has girth at least $5$. Since
$G$ contains at least four odd circuits in $F$, we have at
least $4\cdot 5 + 6 = 26$ vertices in~$G$. Thus $G$ has exactly
$26$ vertices, $C$ has length $6$, and $F$ contains four
circuits of length $5$ plus the circuit $C$. One of the odd
circuits of $F$ is in $F_1$ while three of them are in $F_2$.
Both $G_1'$ and $G_2'$ contain at least two vertices of $C$,
for otherwise $G$ would have a bridge. But then $G_2'$ contains
at least $17$ vertices (three odd $5$-circuits plus at least
two vertices of $C$), which is a contradiction.

We are thus left with the case where $G_1'$ or $G_2'$ has a
\textit{special} edge $e$, that is, an edge that belongs to no
odd circuit of any $2$-factor. By an exhaustive computer search
we have verified that no snark of order at most $16$ has a
special edge. (It is enough to look at snarks without parallel
edges. There is a unique snark on $18$ vertices with a special
edge; it is created by joining the terminals of two copies of
$P_3$ in such a way that the graph is bridgeless.) This
completes the proof.
\end{proof}

The computer search referred to in the proof of
Theorem~\ref{thm:3connected} can be extended to show that there
are exactly two snarks with oddness $4$ on $28$ vertices --
those displayed in Figure~\ref{fig:snark28}.

The following theorem provides an upper bound for the oddness
ratio of  a snark without restriction on connectivity.

\begin{theorem}\label{thm:7.5}
For every even integer $\omega\ge 2$ there exists a snark with
oddness $\omega$ having fewer than $7.5\omega$ vertices.
\end{theorem}

\begin{proof}
Let $G$ be a snark and let $v$ be an arbitrary vertex of $G$.
Construct the graph $G^{(v)}$ by inserting a copy of $P_2$ into
each of the three edges incident with the vertex $v$ of $G$.
This construction increases the number of vertices of $G$ by
$30$ and the oddness by $4$. Indeed, let $F$ be a $2$-factor of
$G^{(v)}$; this $2$-factor has a corresponding $2$-factor $F'$
in $G$ obtained by discarding the circuits of $F$ contained in
the inserted copies of $P_2$ and replacing each path formed by
the traversal of $F$ through a copy of $P_2$ by the appropriate
edge of $G$. The vertex $v$ has degree $2$ in $F$, hence $F$
passes through two edges incident with $v$ in $G^{(v)}$. Each
of the copies of $P_2$ inserted into those edges contains
exactly one $5$-circuit of~$F$. Moreover, the remaining copy of
$P_2$ contains two $5$-circuits of~$F$. The circuit of $F$
passing through $v$ has length increased by $10$ compared to
the corresponding circuit of~$F'$. Put together, $F$ has four
more $5$-circuits than the corresponding $2$-factor $F'$ of
$G$.

Let $R_0$ be the Petersen graph and for each even $i\ge 1$ set
$R_{i+2}=R_i^{(v)}$, where $v$ is an arbitrary vertex of $R_i$.
As explained above, for each integer $t\ge 0$ the graph
$R_{2t}$ has order $30t+10$ and oddness $4t+2$. Thus
$|V(R_{2t})|= 7.5\cdot\omega(R_{2t})-5$.

Further, let $R_1$ be one of the two snarks of order $28$ with
oddness $4$ described in the proof of
Theorem~\ref{thm:3connected}, say $H_1$, and for each odd $i\ge
3$ set $R_{i+2}=R_i^{(v)}$, where $v$ is an arbitrary vertex of
$R_i$. It follows that for each integer $t\ge 0$ the graph
$R_{2t+1}$ has order $30t-2$ and oddness $4t$. Therefore
$|V(R_{2t+1})|=7.5\cdot\omega(R_{2t+1})-2$.

Summing up, for each even integer $\omega\ge 2$ we have
constructed a snark with oddness $\omega$ and order smaller
than $7.5\omega$.
\end{proof}

We can construct many different sequences $(G_i)$ of graphs
with oddness ratio converging to $7.5$ from below. However, no
construction has led to any constant better than $7.5$. This
fact and examples like the family $(R_n)_{n\ge 0}$ from the
previous theorem lead us to conjecture the following.

\begin{conjecture}\label{conj1}
Every snark $G$ with oddness $\omega$ has at least $7.5\omega-5$
vertices.
\end{conjecture}

The smallest snark with oddness $6$ we are aware of is the
graph $R_2$ of order $40$ from the proof of
Theorem~\ref{thm:7.5}; it is determined uniquely up to
isomorphism. Somewhat surprisingly, the oddness ratio of $R_2$
is $20/3$, which is better than the ratio $7$ reached by $R_1$,
one of two smallest graphs with oddness $4$.

A bit more support for Conjecture \ref{conj1} is given by the
following argument. To generalise the construction given in
Theorem~\ref{thm:7.5} we attempt to insert a copy of $P_2$ into
several edges of an appropriate cubic graph $H$. Let $S$ be the
set of edges of $H$ which are replaced by copies of $P_2$. We
have added $10|S|$ vertices and have guaranteed one or two odd
circuits in each copy of $P_2$, depending on whether a
$2$-factor of the resulting graph $H'$ passes through the copy
or not. We show that there is a $2$-factor $F$ passing through
at least $2/3$ of the copies of $P_2$ we have inserted into
$H$. This result is, in fact, due to Fan \cite[Equation~2]{fan}
where it was proved by the technique of light $1$-factors. Our
proof uses the matching polytope instead.

\begin{lemma}[Fan \cite{fan}]\label{lemma:2/3}
For any set $S$ of edges of a bridgeless cubic graph $G$ there
exists a $2$-factor in $G$ which contains at least $2/3\cdot
|S|$ edges from $S$.
\end{lemma}

\begin{proof}
Let $\mathbf{x}$ be a point of the perfect matching polytope
$\mathcal{P}(G)$ such that the sum $$s(\mathbf{x})=\sum_{e\in
E(G)} \mathbf{x}(e)$$ is minimal. The minimum of
$s(\mathbf{x})$ is attained on at least one vertex of
$\mathcal{P}(G)$, hence we may assume that $\mathbf{x}$ is a
vertex. Since $(1/3, 1/3, \dots , 1/3) \in \mathcal{P}(G)$, the
perfect matching corresponding to $\mathbf{x}$ has at most
$1/3\cdot |S|$ edges in $S$. The complementary $2$-factor
therefore contains at least $2/3\cdot |S|$ edges from $S$.
\end{proof}

From Lemma~\ref{lemma:2/3} it follows that the number of added
odd circuits is at most $4/3\cdot |S|$. Since we have added
$10\cdot |S|$ vertices, the oddness ratio of $H$ cannot be
improved to be less than $7.5$, unless we have created many odd
circuits by inserting copies of $P_2$ into even circuits of the
$2$-factor of $H$ corresponding to $F$ -- and many even
circuits in a $2$-factor of $H$ mean that $H$ has a large
oddness ratio.

We do not have any better bounds on resistance than those
established by Steffen \cite{steffen2}: $8 \le A_\rho^2 \le 9$.
Our construction proves that $A_\omega^2\le 7.5$. This shows
that $A_\omega^2$ is strictly smaller than $A_\rho^2$.

\begin{figure}[h]
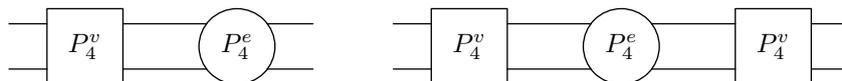

\center
\includegraphics{oddness4-1}\hskip 1cm \includegraphics{oddness4-2}
\caption{Networks $N_1$ (left) and $N_2$ (right)}
\label{fig:4-poles}
\end{figure}

\section{Connectivity 4}

The case of cyclic connectivity $4$ is perhaps most interesting
of all, because snarks with cyclic connectivity smaller than
$4$, or girth smaller than $5$, are usually considered to be
trivial. In a recent paper \cite{hagglund}, Hägglund
constructed an infinite family of cyclically $4$-connected
snarks whose oddness ratio is $15$. The aim of this section is
to present a family of snarks that improves the ratio to~$13$.
We start with the $4$-poles $P_4^v$ and $P_4^e$ defined in
Section~\ref{sec:blocks} and construct two specific
uncolourable $4$-poles $N_1$ and $N_2$. Recall that the
terminals of both $P_4^v$ and $P_4^e$ are partitioned into two
pairs. The colouring properties of $P_4^v$ and $P_4^e$
described in Section~\ref{sec:blocks} imply that if we join one
pair of terminal edges from $P_4^e$ to a pair of terminal edges
from $P_4^v$ we get an uncolourable $4$-pole. This $4$-pole has
18 nonterminal vertices and will be denoted by $N_1$ (see
Figure~\ref{fig:4-poles} left). The $4$-pole $N_2$ arises from
$P_4^e$ and two distinct copies of $P_4^v$ by joining each pair
of terminal edges of $P_4^e$ to a pair of terminal edges in a
different copy of $P_4^v$ (see Figure~\ref{fig:4-poles} right).
Thus $N_2$ has 26 nonterminal vertices. The important property
of $N_2$ is that it is not only uncolourable, but it remains so
even after removing an arbitrary nonterminal vertex. This is
obviously true if the removed vertex $w$ lies in a copy of
$P_4^v$ in $N_2$, because the removal of such a vertex leaves a
copy of $N_1$ in $N_2-w$ intact, and $N_1$ is uncolourable.
Assume that we have removed a vertex $w$ from a copy of $P_4^e$
in $N_2$. Then in every $3$-edge-colouring of $P_4^e-w$ at
least one of the pairs of terminal edges of $P_4^e$ will
receive distinct colours, for otherwise a $3$-edge-colouring of
$P_4^e-w$ would induce a $3$-edge-colouring of $P-w$, which is,
however, uncolourable by the Parity Lemma. On the other hand,
the same pair of edges is forced to have identical colours from
the adjacent copy of $P_4^v$. This again implies that $N_2-w$
is uncolourable.

To construct a cyclically $4$-connected snark with arbitrarily
high oddness we take a number of copies of $N_1$ and a number
of $N_2$, arrange them into a circuit, and join one pair of
terminal edges from each copy to a pair of terminal edges of
its predecessor and another pair of terminal edges to a pair of
terminal edges of its successor. The way in which copies of
$N_1$ and $N_2$ are arranged is not unique, therefore we may
get several non-isomorphic graphs even if we only use copies of
one of $N_1$ and $N_2$.

In this construction, each copy of $N_1$ adds $1$ and each copy
of $N_2$ adds $2$ to the resistance of the resulting graph.
Thus if we take $r$ copies of $N_2$, we get a cyclically
$4$-connected snark of order $26r$ with resistance $2r$ and
oddness at least $2r$. If we take $r$ copies of $N_2$ and one
copy of $N_1$ we get a cyclically $4$-connected snark of order
$26r+18$ with resistance $2r+1$ and oddness at least $2r+2$. In
particular, this shows that $A_\omega^4\le A_\rho^4\le 13$.

For $r=1$ our construction produces a cyclically $4$-connected
snark of order $44$ with resistance $3$ and oddness $4$ (see
Figure~\ref{snark44}). This is currently the smallest known
non-trivial snark of oddness at least $4$. It improves the best
previously known values of $50$ and $46$ (see
\cite{brinkmann,hagglund}).

We have tried various other approaches to construct a
cyclically $4$-connected snark with oddness $4$ on fewer
vertices (mostly computer-assisted), but all of them eventually
led to snarks of order $44$, not necessarily isomorphic to the
one from Figure~\ref{snark44}. We therefore believe that the
following is true.

\begin{conjecture}
The smallest cyclically $4$-connected snark with oddness $4$
has $44$ vertices.
\end{conjecture}

\begin{figure}[h!]
\center
\includegraphics[scale=.8]{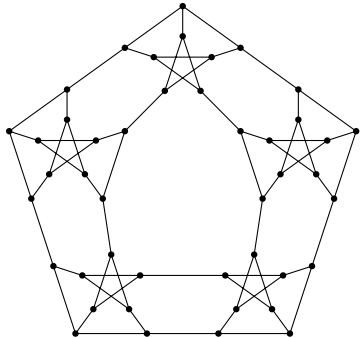}
\caption{A cyclically $4$-connected snark of order $44$ with oddness $4$.}
\label{snark44}
\end{figure}

\section{Connectivity 5}

Steffen in \cite[Theorem 2.3]{steffen2} constructed a
cyclically $5$-connected snark of order $608r$ with oddness at
least $8r$ and resistance $6r$ for every positive integer $r$.
His construction shows that $A_{\rho}^5 \le 101+1/3$ and
$A_{\omega}^5 \le 76$. Here we construct, for each integer $r
\ge 2$, a cyclically $5$-connected snark of order $25r$ with
resistance $r$. Our construction thus yields $A_{\omega}^5 \le
A_{\rho}^5 \le 25$.

We take two disjoint copies of $P_5^{ev}$ and one copy of
$P_5^{vvv}$. Recall that $P_5^{vvv}$ has two pairs of terminal
edges and one single terminal edge. We join one pair of
terminals to the pair of terminals in the first copy of
$P_5^{ev}$ and the other pair to the pair in the second copy of
$P_5^{ev}$ (see Figure~\ref{fig:7pole}). The resulting $7$-pole
$Z$ has $25$ vertices; it coincides with the one constructed by
Steffen \cite[Figure 1]{steffen2}.

In every $3$-edge-colouring of $P_5^{vvv}$ the terminal edges
of one of the two pairs of terminals must have the same colour.
On the other hand, in every $3$-edge-colouring of $P_5^{ev}$
the edges in the pair of terminals have different colours.
Therefore the $7$-pole $Z$ is uncolourable. The required graphs
can be easily constructed by joining terminals from $r$
disjoint copies of $Z$ in such a way that the graph becomes
cyclically $5$-connected. Since $Z$ is uncolourable, we need to
remove a vertex from each copy of $Z$ in order to get a
$3$-edge-colourable subgraph. The resistance of the resulting
graph is therefore at least $r$.

\begin{figure}
\center
\includegraphics{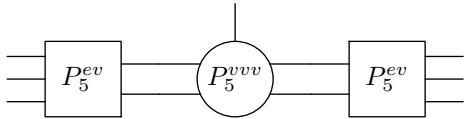}
\caption{A non-colourable cyclically $5$-connected $7$-pole.}
\label{fig:7pole}
\end{figure}

\section{Connectivity 6}

The first infinite class of cyclically $6$-connected snarks
with arbitrarily large oddness was described by Kochol
\cite{kocholOddness}. For each positive integer $n$ he
constructed a cyclically $6$-connected snark of order $118r$
with resistance at least $r$, thereby establishing the bound
$A^6_\omega\le A^6_\rho\le 118$. We improve this upper bound to
$99$ by constructing, for each integer $r\ge 2$, a snark with
resistance at least $r$ and with order $99r$ or $99r+1$,
depending on whether $r$ is even or odd, respectively.

\begin{figure}
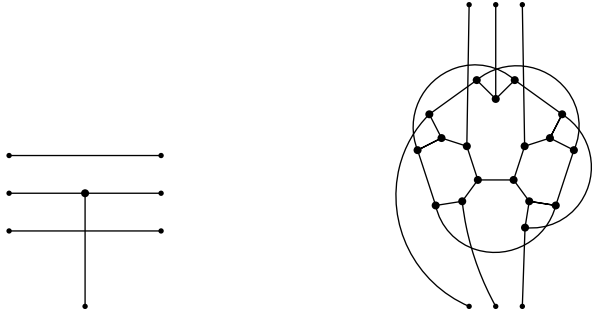

\center
\includegraphics{oddness6-4}\hskip 3 cm\includegraphics{oddness6-3}
\caption{The supervertex $X$ and the superedge $Y$.}
\label{fig:superelements}
\end{figure}

Take $r$ copies of $P_3$, say $Q_1,Q_2,\ldots,Q_r$, where $r\ge
2$. Arrange them into a circuit and, for each
$i\in\{1,2,\ldots, r\}$, join one of the terminals of $Q_i$ to
a terminal of $Q_{i-1}$ and another terminal to a terminal of
$Q_{i+1}$, with indices reduced modulo $r$. This leaves one
terminal of each $Q_i$ unmatched and produces an $r$-pole.
Partition the terminals of this $r$-pole into pairs and, if $r$
is odd, one triple. Identify the terminals of each partition
set, and suppress the resulting $2$-valent vertices to obtain a
cubic graph~$L_r$. Its order is $9r$ or $9r+1$, depending on
whether $r$ is even or odd, respectively. The identification
process can clearly be performed in such a way that $L_r$ is
$3$-connected.

To obtain a snark with cyclic connectivity $6$ we apply
superposition to $L_r$. Nontrivial supervertices will be copies
of the $7$-pole $X$ with a single nonterminal vertex shown in
Figure~\ref{fig:superelements} (left), while nontrivial
superedges will be copies of the $6$-pole $Y$ with $18$
nonterminal vertices created from the flower snark $J_5$ by
splitting off two nonadjacent vertices, one from the single
$5$-circuit of $J_5$; see Figure~\ref{fig:superelements}
(right). We choose a circuit $C$ in $L_r$ which intersects each
copy of $P_3$ in the edges displayed bold in
Figure~\ref{fig:super6} (left) and finish the superposition by
replacing each vertex on $C$ with a copy of $X$ and each edge
on $C$ with a copy of $Y$; we use trivial supervertices and
superedges everywhere outside $C$. The superposition is
indicated in Figure~\ref{fig:super6}.

\begin{figure}
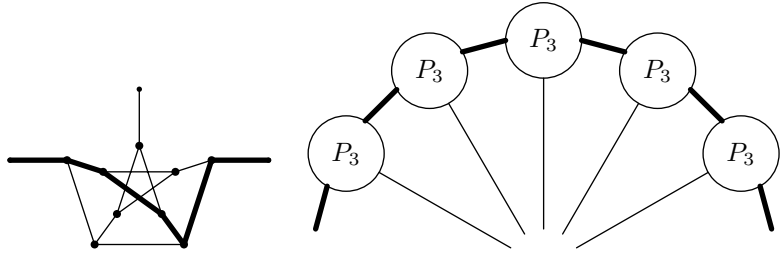

\center
\begin{tabular}{ccc}
\includegraphics{oddness6-1} & \includegraphics{oddness6-2}\\
\end{tabular}
\caption{The superposed circuit passing through a copy of $P_3$
and the arrangement of the copies of $P_3$ in $M_r$.}
\label{fig:super6}
\end{figure}

Let $M_r$ be the resulting graph. The construction replaces the
$5r$ vertices of $C$ with the same number of new vertices, one
for each nontrivial supervertex, and adds $18\cdot 5r=90r$
vertices, $18$ new vertices for each nontrivial superedge. It
follows that $M_r$ has $99r$ or $99r+1$ vertices depending on
whether $n$ is even or odd, respectively. Due to the choice of
supervertices and superedges and the fact that the removal of
the edges of $C$ from each copy of $P_3$ in $L_r$ leaves an
acyclic subgraph, the graph $M_r$ is indeed cyclically
$6$-connected.

We prove that $\rho(M_r)\ge r$. Observe that in order to get a
colourable graph from $L_r$, one has to delete at least one
vertex from each copy of $P_3$, that is, at least $r$ vertices
in total. By Proposition~\ref{prop:superresist},
$\rho(M_r)\ge\rho(L_r)\ge n$, and therefore $\omega(M_r)\ge r$.
This yields the bounds $A_\omega^6 \le A_\rho^6\le 99$.

\bigskip

\noindent\textbf{Acknowledgements.}  The authors acknowledge
partial support from the research grants APVV-0223-10, VEGA
1/0634/09, and from the APVV grant ESF-EC-0009-10 within the
EUROCORES Programme EUROGIGA (project GReGAS) of the European
Science Foundation. The second author acknowledges partial
support also from VEGA 1/1005/12.

\end{document}